\documentclass{amsart}

\usepackage{hyperref}
\hypersetup{colorlinks=true,%
            citecolor=black,%
            filecolor=black,%
            linkcolor=black,%
            urlcolor=black,%
         pdfsubject={},%
        pdfnewwindow=true,%
        pdffitwindow=true,%
        pdftoolbar=true,%
        pdfmenubar=true,%
        bookmarks=true,%
        unicode=false,%
        pdftex}

\usepackage{amsthm,amsmath,amssymb,a4wide}
\numberwithin{equation}{section}

\def\refer#1#2#3#4#5{#1:\ {\sl #2}\ {\bf #3}\ {(#4)}\ #5;\ }
\def\rref#1~{~(\ref{#1})}
\def\bR{\mathbb{R}}

\usepackage{color}
\definecolor{green}{rgb}{0,1,0.5}
\definecolor{darkblue}{rgb}{0,0,0.5}
\def\blu#1~{{{~\color{blue}{#1}}}}

\def\red#1~{{{~\color{red}{#1}}}}
\def\gre#1~{{{~\color{green}{#1}}}}
\def\blu#1~{{~\color{blue}{#1}}}
\def\yel#1~{{{~\color{yellow}{#1}}}}
\def\redb#1~{{{~\colorbox{red}{#1}}}}
\def\greb#1~{{{~\colorbox{green}{#1}}}}
\def\blub#1~{{~\colorbox{blue}{#1}}}
\def\yelb#1~{{~\colorbox{yellow}{#1}}}

\newcommand{\bequ}{\begin{equation}}
\newcommand{\enqu}{\end{equation}}
\newtheorem{lem}{Lemma}

\begin{document}

\title{A radiating spin chain as a model \\ of irreversible dynamics}
\author{P. B\'ona, M. \v Sira\v n}
\thanks{pavel.bona@gmail.com, miso.siran@gmail.com}

\maketitle
\centerline{\Large\it Department of theoretical physics, Comenius University, Bratislava}
\begin{abstract}
\noindent We construct a finite spin-1/2 chain model (quantum domino) interacting with a Fermi field, capable of emitting a scalar fermion from the last spin in the chain. The chain with dynamics gradually reversing the neighbouring spins emits eventually a fermion which escapes then to infinity, and the chain converges to a stationary state. We determine the rate of convergence of the time evolution of the system for $t\rightarrow\infty.$  We prove that the probability of fermion emission as a function of the time $t$ is $1-O(|t|^{-m})$ for arbitrary $m \in \mathbb{N}$. We propose that this fast rate of convergence could serve as an approximate theoretical possibility for the ``effective'' description of the measurement process in the sense proposed by K. Hepp in \cite{Hepp}. This all will be preceded by an outline of explicitly solvable dynamics of infinite version of the spin chain which exhibits transition from a locally perturbed unstable stationary state to a truly macroscopically different - disjoint state, but with slow convergence for $t\rightarrow\infty$.\end{abstract}

\tableofcontents

\section{Motivation and introduction}
 An old and not fully understood problem of physical description of Nature is that of theoretically consistent description of irreversible behavior in the framework of theories invariant with respect to the time inversion. The notorious  ``quantum measurement problem'' (QMP), cf. e.g. \cite{Hepp}, also belongs to this class of problems. It is a known fact that these problems cannot be solved in conceptually and mathematically unambiguous terms in the framework of traditional ``fundamental'' theories of finite systems in finite times, and the QMP is impossible to solve also in the limit of infinite time with the algebra of observables of finite systems, see e.g. \cite{dissert}. It is possible, however, to solve these problems (in a restricted or weakened sense) in the traditional nonrelativistic quantum mechanics (QM) extended to idealized infinitely large systems and infinite times. This was proposed by Klaus Hepp for the description of QMP in \cite{Hepp}. A specific model of this kind was constructed as an infinite spin chain (called the ``quantum domino'' (QD)) \cite{dissert,Bona}, with dynamics described by a one-parameter group of automorphisms (i.e. reversible) with unstable stationary states: After a specific small perturbation of such a state the system evolves (in the limit as time $t\rightarrow\infty$) into another stationary state with changed orientations of infinitely many spins of the chain, hence to a state macroscopically distinguishable from the initial state of this infinite system.\footnote{If that stationary state is perturbed by an arbitrary local perturbation, it evolves in that model also to macroscopically different state which need not be stationary on microscopic scale, cf.\cite{Bona}.} This can be, perhaps, considered as a model for irreversible behavior of the perturbed initial stationary state of the chain, and it was used also for construction of (schematic) models of measurement in QM (cf. e.g. \cite{dissert}, \cite{bon}).

 Our mathematically exactly solvable models are in their physical interpretation ``approximate'' only.
 Let us mention here, however, that during the course of great success of physical theories during the last two hundred years, various approximations (and also some mathematically not controlled ``approximations'') in their formulations nd/or  applications have been accepted as solutions of posed problems. Examples include: thermodynamical limits to infinite systems by dealing with large but finite systems; infinite time limits in description of particle scattering observed experimentally in finite-size  laboratories during short time intervals; conceptually unclear mathematics of relativistic quantum field theory; neglecting environmental influences and/or ascribing (essentially on a base of observer's intuition) some differences between theory and empirical results to such uncontrolled influences, etc. Laboratory experimental confirmation of a theory can be made in some limits of precision only and, moreover, by gradually better isolation of investigated systems from ``external perturbations''. And some {\it ideal precision} test of theoretical results which are expressed e.g. by real numbers cannot be done exactly at all. Approximations are often necessary also from ``purely theoretical'' reasons due to our inability (and sometimes also mathematical impossibility) of precise calculations of mathematically ``existing'' results. For instance, calculated dynamics might be unstable, or we are unable to control the rate of convergence (or the convergence at all) of series of some ``consecutive  approximations'' (perturbation series), etc.

 In any case, the physical theory is made by people, and on behalf of people, who need to use it in their actual variable conditions, and this usually leads to necessity to make approximations. In this connection we also have to keep in mind the words of Niels Bohr concerning QM: {\it There is no quantum world. There is only an abstract quantum physical description. It is wrong to think that the task of physics is to find out how nature is. Physics concerns what we can say about nature.} \cite{Bohr} These considerations could help to justify our motivation and conclusions of this paper.

 A brief exposition of the model of quantum domino will be given in Section \ref{domino}.
 The dynamics of QD is completely explicitly solved not only for infinite chain, but also for arbitrary finite \mbox{(sub-)}chain \cite{Bona}. The language of $C^*-$algebras in the Section  \ref{domino} is used because a macroscopic distinction between states of a system would be difficult to express (in a mathematically correct way) in the language of standard QM.

 In the remaining sections of this paper we shall build a mathematical model of a quantum mechanical system by the use of a finite (sub-\nolinebreak)chain of QD interacting with a (nonrelativistic scalar) Fermi field (which simulates some radiation field) to obtain a model of effectively irreversible behavior  for $t\rightarrow\infty$, with time-inversion symmetric (unitary) dynamics. The $C^*-$algebra language can be easily avoided there.

 Although a ``true irreversibility'' is achieved in our models in weak limits for $t=+\infty$ only, the probability of reversing time evolution of these quantum (sic!) systems after sufficiently long time (by construction of a physically acceptable model of mechanism of such a reversion) after preparing a well determined initial state is extremely  small.
 We could paraphrase Ludwig Boltzmann by: ``Try, please, to reverse empirically this state to obtain eventually the (reversed) initial state!''

\section{The quantum domino}\label{domino}
We shall describe here without details (for which we refer to \cite{Bona}) the model of infinite spin chain which we call quantum domino (QD). The spins are ordered by indices $i\in\mathbb{Z}$ and the Hamiltonian produces a three body interaction for the $i$-th, $(i+1)$-th and $(i+2)$-th spin. This interaction can be described easily as follows: given a state with the $i$-th spin "pointing up" and with the $(i+2)$-th spin "pointing down" the $(i+1)$-th spin changes its orientation to an opposite one. The dynamics of the both sided infinite spin-1/2 quantum chain has spin-configurations ``all spins pointing up'', and ``all spins pointing down'' as stationary states, which are unstable: If we reverse the direction of one of the spins in these states, the state will develop in the limit $t\rightarrow\infty$ into another stationary state, in which all the spins lying on one side of the reversed spin are also reversed, and all the spins lying on the other side of that spin stay unchanged. We shall show in this section how such a dynamics can be obtained.

Let the $C^*-$algebra of observables  $\frak{A}$ be the $C^*-$tensor product of countably infinite copies of algebras\footnote{i.e. the $C^*-$inductive limit of finite products of such algebras} of complex $2\times 2$ matrices generated by the spin creation and annihilation operators $a_j^*, a_j,\ j\in\mathbb{Z}$ satisfying the following (anti)commutation relations
\begin{equation}
[a_i,a_j]=[a^*_i,a_j]=0, \ \  i\neq j \label{anti}
\end{equation}
\begin{equation*}
a_ia_i=0,\ \ a^*_ia_i+a_ia^*_i=1,
\end{equation*}
for all $i, j\in\mathbb{Z}$. The algebra $\mathfrak{A}$ is simple, hence each its nonzero representation is faithful. We shall describe the dynamics in $\frak{A}$ in the ``vacuum'' representation, i.e. in the GNS representation corresponding to the ``vacuum state'' $\omega_0\in \frak{A}^*_{+1}\equiv\mathcal{S}(\frak{A})\ (\equiv$ the ``state space of $\mathfrak{A}$'' consisting of all norm-continuous positive normalized linear functionals on the algebra $\mathfrak{A}$ interpreted physically as expectation functionals on corresponding ``observables'' from $\mathfrak{A}$) that is given by the relation

\begin{equation}\label{vac1}
\omega_0(a_j^*a_j)= 0,\ \ \forall j\in\mathbb Z.
\end{equation}

This state is pure, hence the GNS representation is irreducible. We shall call the spins in this state to be ``pointing down'', to be specific in verbal expression.
Let the cyclic (``vacuum''- in the lattice gas terminology) vector of this representation be denoted by $\Omega_0$, i.e. for all elements $x\in\frak{A}$ it is
\bequ
\omega_0(x)=\langle\Omega_0|x|\Omega_0\rangle,\ \ \forall x\in \mathfrak A.
\enqu

Here, and in the following, we shall denote elements of $\frak A$ and their representatives by operators in the considered irreducible Hilbert space representation by the same symbols. Let us denote this Hilbert space by $\frak F^s$ (according to Fock-spin).

Let us define a ``finite-subchain Hamiltonian''
\begin{equation}\label{HN}
H_{(j,k)}:=\sum_{n=j}^{k-2}a_n^*a_n(a^*_{n+1}+a_{n+1})a_{n+2}a^*_{n+2}.
\end{equation}
Local time evolution automorphisms of $\frak{A}$ are given by
\bequ
\tau^n_t(x):= \exp(itH_{(-n,n)})x\exp(-itH_{(-n,n)}),
\enqu
and the norm limits
\bequ
\tau_t(x):= \lim_{n\rightarrow\infty}\tau^n_t(x)
\enqu
determine the time evolution in $\frak{A}$ (in the ``Heisenberg picture'').

In our vacuum representation this evolution is determined by a selfadjoint Hamiltonian $H$,
\bequ \tau_t(x)= e^{itH}x\ e^{-itH}. \enqu
Here the (unbounded) operator $H$ can be written in the obvious form (its formal definition is given in \cite{Bona})
\begin{equation}\label{HZ}
H:=\sum_{n\in\mathbb Z}a_n^*a_n(a^*_{n+1}+a_{n+1})a_{n+2}a^*_{n+2}.
\end{equation}
This evolution is time-reflexion invariant, but it is not invariant with respect to the space reflexion $n\mapsto -n$.
The operators $g_j:=a_ja_j^*a^*_{j+1}a_{j+1}$ are integrals of motion. One can also prove that the Hilbert space $\frak F^s$ can be decomposed to orthogonal subspaces and on each of them the restriction of the Hamiltonian $H$ is a bounded operator.

Let $X\subset\mathbb N$ be of finite cardinality, and let $\Omega_X:= \prod_{j\in X}a_j^*\ \Omega_0$. The vectors $\Omega_X$ with all finite $X\subset\mathbb N$, with $\Omega_\emptyset:=\Omega_0$ form an orthonormal basis in $\frak F^s$. Each finite $X\subset\mathbb N$ is of the form $Y_1\bigcup Y_2\bigcup\dots\bigcup Y_r$, where all $Y_k\subset\mathbb N$ are finite, mutually disjoint and of the form $\{j_k,j_k+1,j_k+2,\dots,j_k+m_k-1\}$, with $j_{k+1}>j_k+m_k$, i.e the sets $Y_k\subset X$ form mutually separated ``connected islands''. All the vectors $\Omega_X$ are eigenvectors of all the operators $g_j$. For the set $X$ of the just described structure we have
\bequ
g_j\Omega_X=\begin{cases} \Omega_X & \text{for}\ j=j_k-1, \ \ k=1,2,\dots, r \\ 0 & \text{otherwise}.  \end{cases}
\end{equation}

This implies that the time evolution of the vectors $\Omega_X$ conserves the number  of islands,
leaving the initial (``left'') points $j_k$ of each $Y_k, k=1,2,\dots, r$ unchanged (``occupied'', or ``pointing up''), and the places $j_k-1,\ k=2,3,\dots, r$ as well as $j_1-n (n\in\mathbb N)$ remain all the time ``unoccupied'' (i.e. spins are there ``pointing down''). Hence, the subspaces $\frak F^s_{\{{\bf j}\}}$ spanned by all such vectors with fixed $\{{\bf j}\}:= \{j_1,j_2,\dots, j_r\}$ are reduced by the Hamiltonian $H$. Then the space $\frak F^s$ decomposes as

$$\frak F^s = \bigoplus_{\{\bf j\}}\frak F^s_{\{\bf j\}}$$.

 The structure of the Hamiltonian $H$ shows, moreover, that each  $\frak F^s_{\{{\bf j}\}}$ can be written as the tensor product
\bequ
 \frak F^s_{\{{\bf j}\}}=\Omega^0_{(-\infty,j_1)}\otimes \frak F^s_{[j_1,j_2)}\otimes\dots\otimes \frak F^s_{[j_r,+\infty)},
\enqu
where $\Omega^0_{(-\infty,j_1)}$ is one-dimensional space containing the vector with all spins denumbered by $j<j_1$ ``pointing down'', and the spaces $\frak F^s_{[j_k,j_{k+1})}$ are spanned by $j_{k+1}-j_k-1$ vectors corresponding to the ``islands'' $Y_k$ of all permitted lenghts. Here we understand that $j_{r+1}\equiv +\infty$. From the form of the Hamiltonian we see that the time evolution of vectors in $\frak F^s_{\{{\bf j}\}}$ is described by independent evolutions in each $\frak F^s_{[j_k,j_{k+1})}$ by the Hamiltonians $H_{(j_k,j_{k+1})}$, cf. \rref HN~; for more details see \cite{Bona}.

The result of these considerations is that the evolution of general vectors of our representation (hence also the evolution of any states from $\mathcal{S}(\frak{A})$) can be described by two simpler evolutions, namely, the evolutions in finite chains described by vectors $\Omega_Y$ with the islands $Y$ of restricted lenghts, as well as in the Hilbert space spanned by vectors
$\Omega_Y$ with the islands $Y$ of unrestricted lenghts. Because the interaction in our infinite chain is translation invariant, we can describe these two possibilities as\footnote{We shall use the Dirac bra - ket notation in this section for convenience.}

(1) the evolution in the finite-dimensional Hilbert space $\mathfrak H_N$ spanned by the vectors $$|m\rangle := a^*_1a^*_2\dots a^*_m\Omega_0\ (m=1,2,\dots, N)$$ by the unitary evolution $U_N(t):=e^{-itH_N}$ corresponding to the Hamiltonian $H_N:=H_{(1,N+1)}$ from \rref HN~, and

(2) the evolution in the infinite-dimensional Hilbert space space $\mathfrak H_\infty$ spanned by the vectors $$|m\rangle := a^*_1a^*_2\dots a^*_m\Omega_0\ (m=1,2,\dots, \infty)$$ by the unitary evolution $U_\infty(t):=e^{-itH}$ with the Hamiltonian $H:=H_{(1,+\infty)}$.

Let us express these two instances of dynamics by the matrix elements $\langle n|U(t)|m\rangle$ with the corresponding choice of vectors and unitaries. The result can be obtained by explicitly solving the eigenvalue problem for $H_N$ (expressed in terms of the Tshebyshev polynomials of the second kind), and then taking the limit for $N\rightarrow\infty$, cf. \cite{Bona}.

We shall need the following definition:

\bequ\label{JnN}
J^{(N)}_n(\xi):=\frac{i^n}{N+1}\sum_{j=1}^N\exp\left[-i\xi\cos\left(\frac{j\pi}{N+1}\right)\right]\cos\left(n\frac{j\pi}{N+1}\right).
\enqu

This is an integral sum of Sommerfeld integral representation of the Bessel function

\bequ\label{representation}
J_n(\xi)=\frac{i^n}{\pi}\int_0^\pi e^{-i\xi \cos \alpha}\cos(n\alpha) \textrm{d}\alpha.
\enqu

We can now write the desired expression for the Green function of finite chain:

\bequ\label{JN}
\langle n|U_N(t)|m\rangle = (-i)^{n-m} J^{(N)}_{n-m}(2t) - (-i)^{n+m}J^{(N)}_{n+m}(2t),
\enqu

and


\bequ\label{J}
\langle n|U_\infty(t)|m\rangle = (-i)^{n-m} J_{n-m}(2t) - (-i)^{n+m}J_{n+m}(2t).
\enqu
Let us now take the local perturbation $\omega_1(x):=\omega_0(a_1xa^*_1)\ (x\in\mathfrak{A})$ of the time-invariant vacuum state $\omega_0$. The state $\omega_1$ describes the infinite spin-chain in the state, where all the spins except of one lying in the site $j=1$ are pointing down. Its time evolution $\omega_1(\tau_t(x))\equiv \omega^t_1(x)$ can be expressed in terms of the above written results. Let us, for example, calculate the expectation of ``flipping up'' the spin in the $j-$th place at time $t$. We have

\bequ
\omega_1^t(a^*_ja_j)=\sum_{m=1}^\infty \langle 1|e^{itH}a_j^*a_j|m\rangle \langle m|e^{-itH}|1\rangle =
\sum_{m=j}^\infty \langle 1|e^{itH}|m\rangle \langle m|e^{-itH}|1\rangle = 1-\sum_{m=1}^{j-1}|\langle m|e^{-itH}|1\rangle|^2,
\enqu
since
\bequ
a_j^*a_j|m\rangle = \begin{cases} 0 & (m<j),\\ |m\rangle & (m\geq j),     \end{cases}
\enqu
and the set of vectors $\{ |m\rangle : m\in\mathbb N\}$ forms in the relevant Hilbert space an orthonormal basis.
From \rref J~ and from the recurrent formula for Bessel functions
$$ J_{p+1}(\xi)+J_{p-1}(\xi)=\frac{2p}{\xi}J_p(\xi), $$
 we obtain

\bequ
\omega_1^t(a^*_ja_j)=1-\sum_{m=1}^{j-1}\left[\frac{m}{t} J_m(2t)\right]^2.
\enqu
Because of the asymptotic behavior of the Bessel function for large real arguments, given by $J_p(\xi)=O(|\xi|^{-\frac{1}{2}})$, we obtain asymptotic behavior of our expectation:
\bequ\label{converg}
\omega_1^t(a^*_ja_j) \asymp 1-\frac{{\text const.}}{|t^3|}\ \ (\forall j\in\mathbb N), \ {\text for}\ t\rightarrow\infty.
\enqu
Hence the local perturbation of the state ``all spins are pointing down'' converges according to \rref converg~ to the state ``all spins sitting in sites with $j>0$ are pointing up''.

This can be used in construction of models of quantum measurement. For instance, let the infinite chain without the spin in the site $j=0$ serve as an (model of) ``apparatus'' and the spin at $j=0$ serves as a (model of) ``measured microsystem''. If the apparatus is initially in the state $\omega_\downarrow$ with all spins pointing down, and the measured spin is in a superposition $|initial\rangle= c_\downarrow |\downarrow\rangle + c_\uparrow |\uparrow\rangle$ then the final state of the chain (at $t=\infty$) will be (as a state on the algebra $\mathfrak A$ of the composed system ``measured system + apparatus'') in incoherent genuine mixture $\omega_f$ according to the above described  dynamics: $\omega_f=|c_\downarrow|^2 \omega_0 + |c_\uparrow|^2 \omega_\uparrow$, where the state $\omega_\uparrow$ means that all spins of the composed system lying in sites $j\geq 0$ are pointing up. The states $\omega_0$ and $\omega_\uparrow$ on $\mathfrak A$ are mutually disjoint, what is interpreted here as ``macroscopic difference''. Also, the states $\omega_0$ and $\omega_\uparrow$ define two representations of the algebra of quasi-local observables (see \cite{Bra}  for details) which are not unitary equivalent, and can be distinguished by a measurement of a macroscopic observable. This is an example in the spirit of the models proposed in the classical paper by Hepp \cite{Hepp}.

\section{Model of finite radiating chain}\label{model}
Let us consider now a composite model consisting of a finite spin-1/2 chain of the lenght $N+1$ (indexed by $i=0,\ldots,N$, note we use the expression $i-th$ spin in the sense of this labeling, so the chain begins with the zeroth spin) with the dynamics of QD (as it was  described in the section \ref{domino}), of a two-level unstable ``particle'' modeled by the $N$-th spin-1/2 at the end of the chain, and of a scalar Fermi field interacting with the ``particle''.

The dynamics of the model can be described as follows. The initial (stationary) state consists of the spin chain with all spins pointing down, and the Fermi field in the vacuum state. Flipping of the first spin (either ``by hand'', or by an external influence, e.g. by a scattering with another system) leads to the ``domino effect'' described above, which subsequently flips all the next spins of the finite chain. The flipping of the final spin (interacting with the Fermi field) results in the emission of a fermion by the chain. Since the spin chain interacts with the Fermi field by short range interaction only, the fermion with sufficiently large kinetic energy escapes irreversibly into infinity and the chain remains in the state with all the $N+1$ spins ``pointing up''. This model was proposed by the first author and mentioned in \cite[Section 2.2]{Bon}.\\

Such a model could be considered also as an example of the often debated process of ``decoherence'', where the Fermi field and also a part of the long spin chain can play the role of the ``environment''; it can be used then as a corresponding model of quantum measurement (the ``measured'' object would be then the initial spin of the chain, and the ``measuring apparatus'' is the rest of the finite spin chain). It still allows, however, interference of approximately ``macroscopically different states'', which, with the growing size $N$ of the ``apparatus'', becomes apparently less probable, cf. \cite{leggett}. \footnote{Some process of this ``decoherence type'' seems to be only possibility of description of QMP in framework of conventional physics without modifying QM, cf. \cite{Hepp,leggett}.} We shall return briefly to this aspect of our model in Section \ref{conclusion}. \\

\noindent The Hilbert space of states is $\mathcal{H}:=(\mathbb{C}^2)^{N+1}\otimes \mathfrak{F}$, where $\mathfrak{F}$ is the Fermi Fock space with the vacuum vector $\Omega_0^f$ - the representation space of the $CAR$\ $C^*$-algebra $\mathcal{A}^f$ . For the spin space $(\mathbb{C}^{2})^{N+1}$ we define $\Omega_0^s$ as the ``vacuum'' vector (it is the state of all spins pointing down).

The spin-1/2 creation and annihilation operators $a^*_j,\ a_j\ (j=0,1,\dots, N)$ satisfying \rref anti~ for $i,j=0,1,\ldots, N$ acting on the space $(\mathbb{C}^2)^{N+1}$ generate the (finite dimensional) algebra of spin observables $\mathcal{A}^s$. The vacuum state $\omega_0 \in \mathcal{S}(\mathcal{A}^s)\subset {\mathcal{A}^s}^*$ ($\mathcal{S}(\mathcal{A}^s)$ denotes the set of states on $\mathcal{A}^s$), $\omega_0(a^*_ia_i)=0\ (\forall i)$ defines the cyclic vacuum vector $\Omega_0^s$, such that $\omega_0(a)=(\Omega_0^s,\ a\ \Omega_0^s)$, $a \in \mathcal{A}^s$. The self-adjoint Hamiltonian $H$ on $\mathcal{H}$ is defined as a sum $H:=H_0+V$ (note that the spin chain emits a fermion, so that all creation and annihilation operators are bounded, see \cite[Proposition 5.2.2.]{Bra}, ). The first operator is given by the following expression
\begin{equation}
H_0:=\left(\sum_{n=0}^{N-2}a^*_na_n(a^*_{n+1}+a_{n+1})a_{n+2}a^*_{n+2}-\varepsilon_0a^*_Na_N\right) \otimes I^f+I^s\otimes d\Gamma(h), \label{H0}
\end{equation}
where $\varepsilon_0>0$, and $h$ is a self-adjoint operator on $L^2(\mathbb{R}^3,d^3x)$ that operates under the Fourier transform $\mathcal{F}$ as (for arbitrary $\phi \in L^2(\mathbb{R}^3,d^3x)$)
\begin{equation}\label{ft}
\mathcal{F}[h\phi](\vec{p}):=\varepsilon(\vec{p})\mathcal{F}[\phi](\vec{p}),
\end{equation}
where
\begin{equation*}
\mathcal{F}[\phi](\vec{p}):=\frac{1}{(2\pi)^{3/2}}\int_{\mathbb{R}^3}e^{-i\vec{p}.\vec{x}}\phi(\vec{x})d^3x,
\end{equation*}
with $\varepsilon:\mathbb{R}^3 \to \mathbb{R}_+$ is to be specified later, and $d\Gamma(h)$ is the second quantization of $h$ (see \cite[Sec.5.2.1.]{Bra}). The interaction part of the Hamiltonian is
\begin{equation}
V:=v^2\left(a^*_{N-1}a_{N-1}a^*_N\otimes b^*(\sigma)+a^*_{N-1}a_{N-1}a_N\otimes b(\sigma)\right); \label{V}
\end{equation}
here $v\in \mathbb{R}$ and $\sigma \in L^2(\mathbb{R}^3,d^3x)$ are to be specified later, the (anti-)linear mappings $b, b^*: L^2(\mathbb{R}^3,d^3x) \to \mathcal{L}(\mathfrak{F})$ give the annihilation, resp. creation operators of the scalar fermion satisfying the relations:
\[ \begin{array}{ll}\label{antiF}
  b(\varphi)^2=0, \
  b(\varphi)b^*(\psi)+b^*(\psi)b(\varphi)=(\varphi,\psi)I^f, &\  (\forall \varphi,\psi\in L^2(\mathbb{R}^3,d^3x)).
  \end{array}\]
 We denote by $H^{(N)}$ the part of $H_0$ for the spin chain
\begin{equation*}
H^{(N)}:=\sum_{n=0}^{N-2}a_n^*a_n(a^*_{n+1}+a_{n+1})a_{n+2}a^*_{n+2}\otimes I^f.
\end{equation*}

Define the closed subspace $\mathcal{H}_1\subset\mathcal{H}$ by the following closure
\begin{equation}
\mathcal{H}_1:=\overline{\text{span}\left\{\Omega_0,\beta_n,\beta_N(\phi); n=0,\ldots,N-1,\phi\in L^2(\mathbb{R}^3,d^3x)\right\}}, \label{H1}
\end{equation}
\begin{equation*}
\Omega_0:=\Omega_0^s \otimes \Omega_0^f,
\end{equation*}
\begin{equation*}
\beta_n:=a_0^*\ldots a_n^*\Omega_0^s\otimes \Omega_0^f,
\end{equation*}
\begin{equation*}
\beta_N(\phi):=a_0^*\ldots a_N^*\Omega_0^s\otimes b^*(\phi)\Omega_0^f.
\end{equation*}
Our examination of dynamics of this model will be restricted to the subspace $\mathcal{H}_1$ mainly. This will be possible because of the following Lemma.
\begin{lem}
 The subspace $\mathcal{H}_1$ defined by $(\ref{H1})$ is $H$-invariant, that is $H\mathcal{H}_1\subset \mathcal{H}_1$.\hfill $\triangle$ \end{lem}
 \begin{proof}
For $n=0,\ldots, N-1$,
\begin{equation*}
H_0\beta_n=\underbrace{\left(\sum_{k=0}^{N-2}a_k^*a_k(a^*_{k+1}+a_{k+1})a_{k+2}a^*_{k+2}-\varepsilon_0a^*_Na_N\right)a_0^*\ldots a_n^*\Omega_0^s\otimes\Omega^f}_{\in\mathcal{H}_1}+a_0^*\ldots a_n^*\Omega_0^s\otimes \underbrace{d\Gamma(h)\Omega_0^f}_{=0},
\end{equation*}
and for $n=0, \ldots, N-2$, we have
\begin{equation*}
V\beta_n=v^2\left(a^*_{N-1}a_{N-1}a_N^*a_0^*\ldots a_n^*\Omega_0^s\otimes b^*(\sigma)\Omega_0^f+a_{N-1}^*a_{N-1}a_Na_0^*\ldots a^*_n\Omega_0^s\otimes \underbrace{b(\sigma)\Omega_0^f}_{=0}\right)=
\end{equation*}
\begin{equation*}
=v^2a^*_{N-1}a^*_Na^*_0\ldots a^*_n\underbrace{a_{N-1}\Omega_0^s}_{=0}\otimes b^*(\sigma)\Omega_0^f=0.
\end{equation*}
For the remaining cases
\begin{equation*}
V\beta_{N-1}=v^2a^*_{N-1}a_{N-1}a^*_{N}a_0^*\ldots a_{N-1}^*\Omega_0^s\otimes b^*(\sigma)\Omega_0^f=v^2\beta_N(\sigma) \in \mathcal{H}_1,
\end{equation*}
\begin{equation*}
\left(I^s\otimes d\Gamma(h)\right)\beta_N(\phi)=a_0^*\ldots a_N^*\Omega_0^s\otimes d\Gamma(h)b^*(\phi)\Omega_0^f=\beta_N(h\phi)\in \mathcal{H}_1,
\end{equation*}
\begin{equation*}
(H^{(N)}-\varepsilon_0a^*_Na_N\otimes I^f)\beta_N(\sigma)=-\varepsilon_0\beta_N(\sigma) \in \mathcal{H}_1,
\end{equation*}
\begin{equation*}
V\beta_N(\phi)=v^2\left(a^*_{N-1}a_{N-1}a^*_Na^*_0\ldots a^*_N\Omega_0^s\otimes b^*(\sigma)b^*(\phi)\Omega_0^f+a^*_{N-1}a_{N-1}a_Na^*_0\ldots a^*_N\Omega_0^s\otimes b(\sigma)b^*(\phi)\Omega_0^f\right)=
\end{equation*}
\begin{equation*}
=v^2 a^*_0\ldots a^*_{N-1}\Omega_0^s\otimes b(\sigma)b^*(\phi)\Omega_0^f=v^2(\sigma,\phi)_{L^2(\mathbb{R}^3,d^3x)}\beta_{N-1}\in \mathcal{H}_1.
\end{equation*}
\newline
In the last equation $(\cdot,\cdot)_{L^2(\mathbb{R}^3,d^3x)}$ denotes the standard scalar product in $L^2(\mathbb{R}^3,d^3x)$.
\end{proof}
Hence the subspace $\mathcal{H}_1$ is $H$-invariant, and we can restrict our examination of the dynamics to this subspace.

\section{Dynamics of the model}\label{dynamics}

The dynamics of the model is chosen such that the last $N$-th spin in the chain can emit or absorb a fermion in the state described by the vector $\sigma \in L^2(\mathbb{R}^3,d^3x)$, where emission is connected with ``switching up'' and absorbtion with ``switching down'' of the $N$-th spin. This process is possible if the $(N-1)$-st spin is ``pointing up'' only. This means that if the initial state is described by a vector in $\mathcal{H}_1$, then the probability of the emission of a fermion by the chain (being in $\mathcal{H}_1$) equals to the probability of the $N$-th spin to be ``pointing up''.

We intend to prove that the probability of the emission of a fermion by the chain starting in the initial state of the system corresponding to the vector $\beta_n,\ (n=0,1,\dots, N-1)$ evolves with time $t\to\infty$ quickly to 1, the speed of this convergence being ``almost exponential''. More precisely, we shall show that for specific ``conveniently chosen'' parameters of the model, the probability of the state with ``all the spins pointing up'' is

\bequ\label{evolution}
(\beta_n,e^{itH}a^*_Na_Ne^{-itH}\beta_n) = 1-o(t^{-m}),
\enqu
for all $m\in\mathbb N$ and $t\to+\infty.$ \\

\noindent It will be convenient for analysis of matrix elements of the unitary group $\exp(-itH)$ of time evolution to analyze instead their Fourier transform. Because we do not know a priori nothing about integrability (and the existence of classical Fourier transform) of the functions  $t\to (\psi,\exp(-itH)\varphi)$, we shall consider them as tempered distributions  defined by locally integrable functions. The classical Fourier transform can be extended to such distributions. Effective theorems  are valid for Fourier transforms of distributions defined on half-line (or, in more dimensions, on a cone). The desired Fourier transforms of such distributions in real domain are boundary values on $\mathbb{R}$ of functions complex analytic on an open half-plane, cf. e.g. \cite[Theorem IX.16]{RS2}. We shall obtain in this way, after a choice of model parameters, Fourier transforms of functions of the desired behavior.

Before we begin our analysis of the model, let us prove some general propositions which we will apply later.
 Denote by $R_A$ the resolvent of a closed operator $A$ on the domain $D(A)\subset \mathcal{H} $, that is, $R_A(\xi)=(A-\xi I)^{-1}$, with $\xi \notin \sigma(A)$, where $\sigma(A)$ denotes the spectrum of $A$. We will use the Heaviside function $\theta$ defined as
\begin{equation*}
\theta(t):= \left\{
  \begin{array}{l l}
    0 & t<0,\\
    1 & t\geq 0.\\
  \end{array} \right.
\end{equation*}

\begin{lem}
\label{fourier}
Let $e^{-itH}$ be any (unitary) time evolution group. Then the Fourier transform of its (truncated) matrix elements for given $\phi, \psi \in \mathcal{H}$ is
\begin{equation*}
\mathcal{F}[\theta(t)(\phi,e^{itH}\psi)](\xi)=\frac{i}{\sqrt{2\pi}}(\phi,R_{H}(\xi)\psi),
\end{equation*}
for all  $\xi\notin \sigma(H)$, $\text{Im}\ \xi < 0$.\hfill $\triangle$
\end{lem}
\begin{proof}
We shall calculate the Fourier transform defined as
\begin{equation*}
\mathcal{F}[\theta(t)(\phi,e^{itH}\psi)](\xi)=\frac{1}{\sqrt{2\pi}}\int_{\mathbb{R}_+}dt e^{-i\xi t} (\phi,e^{itH}\psi).
\end{equation*}
The spectral theorem for the hamiltonian (see \cite[Theorem VIII.6]{RS}),
\begin{equation*}
H=\int_{\sigma(H)}\lambda E_H(d\lambda)
\end{equation*}
leads by the functional calculus to
\begin{equation*}
f(H):=\int_{\sigma(H)} f(\lambda) E_H(d\lambda),
\end{equation*}
where $E_H: \mathcal{B} \to \mathcal{L}(\mathcal{H})$ is the unique projector valued measure from the aforementioned von Neumann theorem ($\mathcal{B}$ denotes the Borel $\sigma$-algebra of subsets of $\mathbb{R}$), and $f$  is continuous on the spectrum, $f \in C(\sigma(H))$. If we choose $f(\lambda):=e^{it\lambda}$, we can express the Fourier transform of the (truncated) matrix elements $\theta(t)(\phi,e^{itH}\psi)$ in the form:
\begin{equation}
\frac{1}{\sqrt{2\pi}}\int_{\mathbb{R}_+}dt e^{-it\xi} (\phi,e^{itH}\psi)=\frac{1}{\sqrt{2\pi}}\int_{\mathbb{R}_+}dt e^{-it\xi}\int_{\sigma(H)}e^{it\lambda}(\phi,E_H(d\lambda)\psi). \label{pi}
\end{equation}
Integrating with respect to $t$ in the case Im $\xi < 0$ we obtain
\begin{equation*}
\mathcal{F}[\theta(t)(\phi,e^{itH}\psi)](\xi)=
\frac{1}{\sqrt{2\pi}}\int_{\sigma(H)}\left.\frac{e^{it(\lambda-\xi)}}{i(\lambda-\xi)}\right|^\infty_0(\phi,E_H(d\lambda)\psi),
\end{equation*}
and because
\begin{equation*}
\int_{\sigma(H)}\frac{1}{\lambda-\xi}(\phi,E_H(d\lambda)\psi)= (\phi,R_{H}(\xi)\psi),
\end{equation*}
  we have the result.
\end{proof}
\noindent Another useful result is that we obtain the resolvent $R_H(\lambda)$ as a solution of an operator equation. \begin{lem}
Suppose $H=H_0+V\in \mathcal{L}(\mathcal{H})$ and $\xi \notin \sigma(H)\cup\sigma(H_0)$. Then the resolvent $R_H(\xi)$ is the solution of the operator equation \label{res}
\begin{equation}
R_H(\xi)=R_{H_0}(\xi)(I-VR_{H}(\xi)).\label{3.2}
\end{equation}
Hence, the Fourier transform of the (truncated) matrix elements of the time evolution operator for Im $\xi < 0$ is given by:
\begin{equation*}
\mathcal{F}[\theta(t)(\phi,e^{itH}\psi)](\xi)=\frac{i}{\sqrt{2\pi}}(\phi,R_{H_0}(\xi)\psi)-\frac{i}{\sqrt{2\pi}}(\phi,R_{H_0}(\xi)VR_{H}(\xi)\psi).
\end{equation*}\hfill $\triangle$
\end{lem}
\begin{proof} The right-hand side of equation ($\ref{3.2}$) can be expressed as
\begin{equation*}
R_{H_0}(\xi)\left((H-\xi)R_H(\xi)-VR_H(\xi)\right)=R_{H_0}(\xi)(H-\xi-V)R_H(\xi)=
\end{equation*}
\begin{equation*}
=R_{H_0}(\xi)(H_0-\xi)R_H(\xi)=R_H(\xi),
\end{equation*}
which gives ($\ref{3.2}$). The rest of this Lemma is valid due to Lemma $\ref{fourier}$.
\end{proof}

\noindent After these general considerations we shall proceed to examine the dynamics of our model.
\noindent Now, the goal is to set $\sigma \in L^2(\mathbb{R}^3,d^3x)$, $\varepsilon_0, v$ and $\varepsilon(\vec{p})$ in the model $(\ref{H0}), (\ref{ft}), (\ref{V})$ such that for $n=0,\ldots, N-1$ we have
\begin{equation}
\lim_{t \to \infty}(\beta_n,e^{itH}a^*_Na_Ne^{-itH}\beta_n)=1, \label{limit}
\end{equation}
where $\beta_n:=a^*_0\ldots a^*_n\Omega_s \otimes \Omega_f$ are the states of the first $n+1$ spins pointing up.  The matrix element in \rref limit~ measures probability of emission of the fermion in time $t$ if the system was in the state $\beta_n$ at $t=0.$   We want to calibrate the model in such a way that \rref limit~ will be satisfied. In this case, equation(s) ($\ref{limit}$) states that in the limit of infinite time, starting in $t=0$ from any of the states $\beta_n\in \mathcal{H}_1$,  the fermion escapes into infinity, and all the spins of the chain remain pointing up ($\beta_N$ is a stationary state). \\

\noindent Central to our analysis is the matrix element
\bequ\label{F-mn}
F_{mn}:=(\beta_m,R_H(\xi),\beta_n),
\enqu
since by Lemma $\ref{fourier}$  we obtain
\begin{equation}
\mathcal{F}[\theta(t)(\beta_m,e^{itH}\beta_n)](\xi)=\frac{i}{\sqrt{2\pi}}F_{mn}(\xi).
\end{equation}
\noindent
We will show that this matrix element can be expressed with the use of the following matrix elements:
\begin{subequations}
\bequ\label{f-mn}
f_{mn}(\xi):=(\beta_m,R_{H_0}(\xi)\beta_n),\ m,n=0,\ldots, N-1,
\enqu
and
\bequ\label{f-sig}
f^\sigma_{N N}(\xi):=(\beta_N(\sigma),R_{H_0}(\xi)\beta_N(\sigma)).
\enqu
\end{subequations}
 This will prove useful in the analysis of the singularities of $F_{mn}$ as the elements $f_{mn}$ can be expressed using the integral sums of the Sommerfeld representation of Bessel functions (\ref{JnN}). Note that the capital letter $F$ is used for matrix elements of the resolvent of $H$ and the lowercase letter $f$ is used for matrix elements of the resolvent of $H_0$. Also note the fact that for $n=0,\ldots,N-1$\ it is $R_{H_0}(\xi)\beta_n=R_{H^{(N)}}(\xi)\beta_n$ so that we can exchange $H^{(N)}$ for $H_0$ in the definitions of $f_{mn}$.\\

\noindent Exploiting the identity operator trick we have for the space $\mathcal{H}_1$
\begin{equation}\label{sum}
VR_H(\xi)\psi=\sum_{\substack{
\phi\in\{\Omega_0,\beta_n,\beta_N(\phi_j); \\
n=0,\ldots,N-1,j=0,1,\dots\}}}
(\phi,R_H(\xi)\psi)V\phi \ ;
\end{equation}
 by setting $\phi_j$ such that $\phi_0:=\sigma$ with $(\sigma,\sigma)_{L^2(\mathbb{R}^3,d^3x)}=1$ and $\{\phi_j; j=0,1,...\}$ forms an orthonormal basis of $L^2(\mathbb{R}^3,d^3x)$, then $V\beta_{N-1}=v^2\beta_N(\sigma)$ and $V\beta_N(\phi_0)=v^2\beta_{N-1}$ are the only non-zero $V\phi$ terms. Recalling that $\text{span}\ \{\beta_n, n=0, \ldots, N-1\}$ is $H_0$-invariant and that it is orthogonal to the states with all spins up, that is $(\beta_n,\beta_N(\phi))=0$, we have
\begin{equation}\label{nonzero}
(\beta_n,R_{H_0}(\xi)\beta_N(\sigma))=0.
\end{equation}

From Lemma $\ref{res}$, we have
\begin{equation}
\begin{split}
&F_{mn}(\xi)=(\beta_m,R_{H_0}(\xi)\beta_n)-(\beta_m,R_{H_0}(\xi)VR_H(\xi)\beta_n)\\
&=(\beta_m,R_{H_0}(\xi)\beta_n){ -\sum}_{\substack{
\phi\in\{\Omega_0,\beta_n,\beta_N(\phi_j); \\
n=0,\ldots,N-1,j=0,1,\dots\}}}(\beta_m,R_{H_0}(\xi)V\phi)(\phi,R_H(\xi)\beta_n)\\
&=(\beta_m,R_{H_0}(\xi)\beta_n)-v^2(\beta_m,R_{H_0}(\xi)\beta_{N-1})(\beta_N(\sigma),R_H(\xi)\beta_n)\\
&=f_{mn}(\xi)-v^2f_{m\ N-1}(\xi)(\beta_N(\sigma),R_H(\xi)\beta_n). \label{matrix2}
\end{split}
\end{equation}
Here we used equation (\ref{nonzero}) which implies that only $\phi=\beta_N(\sigma)$ gives a non-zero term in the sum. Similarly, the matrix element $(\beta_N(\sigma),R_H(\xi)\beta_n)$ on the right hand side of the last equation gives, using Lemma $\ref{res}$, (\ref{sum}) and (\ref{nonzero}),
\begin{equation}
\begin{split}\label{phi}
&\ \ \ \ \ (\beta_N(\sigma),R_H(\xi)\beta_n)=-(\beta_N(\sigma),R_{H_0}(\xi)VR_H(\xi)\beta_n)=\\
&\ ={ -\sum}_{\substack{
\phi\in\{\Omega_0,\beta_n,\beta_N(\phi_j); \\
n=0,\ldots,N-1,j=0,1,\dots\}}}\left(\beta_N(\sigma),R_{H_0}(\xi)V\phi\right)(\phi,R_H(\xi)\beta_n)=\\
&=-v^2(\beta_N(\sigma),R_{H_0}(\xi)\beta_N(\sigma))(\beta_{N-1},R_H(\xi)\beta_n)=-v^2f^\sigma_{N N}(\xi)F_{N-1\ n}(\xi),
\end{split}
\end{equation}
as only $\phi=\beta_{N-1}$ gives a non-zero term in the sum. The matrix element $F_{N-1\ n}$ can be expressed from  (\ref{matrix2}) with $m=N-1$ and using (\ref{phi}) as
\begin{equation*}
F_{N-1\ n}(\xi)=f_{N-1\ n}(\xi)+v^4 f_{N-1 N-1}(\xi)f^\sigma_{N N}(\xi)F_{N-1\ n}(\xi).
\end{equation*}
From this equation we have the desired expression for $F_{mn}$, $0\leq m,n\leq N-1$ in terms of $f_{mn}$ and $f^\sigma_{N N}$
\begin{equation}\label{Fmn}
F_{mn}(\xi)=f_{mn}(\xi)+\frac{v^4f_{m\ N-1}(\xi)f^\sigma_{N N}(\xi)f_{N-1\ n}(\xi)}{1-v^4f_{N-1\ N-1}(\xi)f^\sigma_{N N}(\xi)}.
\end{equation}
\\

Let us now examine the matrix elements $f_{m n}$ and $f^\sigma_{N N}$ using the Fourier transform and the integral sums of the Sommerfeld representation of Bessel functions (\ref{representation}). First, observe that the matrix elements $f_{mn}$ are symmetric in the sense that
\begin{equation*}
f_{mn}(\xi)=(\beta_m,R_{H^{(N)}}(\xi)\beta_n)=\frac{1}{i}\int_0^\infty e^{-it\xi}(\beta_m,e^{itH^{(N)}}\beta_n)dt=f_{nm}(\xi),
\end{equation*}
(cf. \cite[(35)]{Bona}). By Lemma $\ref{fourier}$, $f_{mn}$ can be expressed in the form
\begin{equation*}
f_{m n}(\xi)=\frac{\sqrt{2\pi}}{i}\mathcal{F}\left[\theta(t)(\beta_m,e^{itH_0}\beta_n)\right](\xi).
\end{equation*}
Using the fact that $e^{itH_0}\beta_n=e^{itH^{(N)}}\beta_n$ and expressing the Fourier transform we then obtain
\begin{equation*}
f_{m n}(\xi)=\frac{1}{i}\int_0^\infty e^{-it\xi}(\beta_m,e^{itH^{(N)}}\beta_n)dt.
\end{equation*}
According to (\ref{JN}), $f_{m n}$ can be expressed in the form (let us stress that the infinite spin chain of section \ref{domino} started with the spin $j=1$, whereas the spin chain of this section starts with the zeroth spin)
\begin{equation*}
f_{m n}(\xi)=\frac{1}{i}\int_0^\infty e^{-it\xi}[(-i)^{m-n}J_{m-n}^{(N)}(-2t)-(-i)^{m+n+2}J_{m+n+2}^{(N)}(-2t)]dt,\
 m,n=0,1,\dots,N-1.
\end{equation*}
where $J^{(N)}_n$ are the integral sums of the Sommerfeld representation of Bessel functions (\ref{JnN}).
Since
\begin{equation}
\int_0^\infty e^{-it\xi}e^{it2\cos \alpha}dt=\left.\frac{e^{-it\xi+it2\cos \alpha}}{-i\xi+i2\cos\alpha}\right|^\infty_0=-\frac{i}{\xi-2\cos\alpha}, \label{cos}
\end{equation}
(note that $\text{Im}\ \xi < 0$), we have the following representation of $f_{m n}:$
\begin{equation*}
f_{mn}(\xi)=\frac{1}{N+1}\frac{1}{i}\int_0^\infty e^{-it\xi}\sum_{j=1}^N \left[\cos\left((m-n)\frac{j\pi}{N+1}\right)-\cos\left((m+n+2)\frac{j\pi}{N+1}\right)\right]e^{i2t\cos\left(\frac{j\pi}{N+1}\right)}dt
\end{equation*}
\begin{equation*}
=-\frac{1}{N+1}\sum_{j=1}^N\frac{2\sin\left(\frac{(m+1)j\pi}{N+1}\right)\sin\left(\frac{(n+1)j\pi}{N+1}\right)}{\xi-2\cos\left(\frac{j\pi}{N+1}\right)}.
\end{equation*}
To simplify the formulae, we write
\begin{equation}\label{anm}
a_j^{(n,m)}:=-\frac{2}{N+1}\sin\left(\frac{(n+1)j\pi}{N+1}\right)\sin\left(\frac{(m+1)j\pi}{N+1}\right), \ \ j=1,\ldots, N,\ \ n, m=0,\ldots, N-1,
\end{equation}
and the eigenvalues of $H_N$
\begin{equation}\label{xj}
x_j:=2\cos\left(\frac{j\pi}{N+1}\right), \ \ j=1,\ldots,N.
\end{equation}
With this notation, we express the functions $f_{nm}$ as a sum
\begin{equation}\label{fmn}
f_{nm}(\xi)=\sum_{j=1}^N\frac{a_j^{(n,m)}}{\xi-x_j},
\end{equation}
where we have $x_j \neq x_k\ (j \neq k).$\\

\begin{lem}\label{s}
With the notation \rref f-sig~ and \rref F-mn~, if $f^\sigma_{N N}(x_j):=\displaystyle\lim_{\nu \to 0^+}f^\sigma_{N N}(x_j-i\nu)$
 is nonzero and finite, then $F_{nm}$ has no singularity in $x_j$ for $v\neq 0$.\hfill $\triangle$
\end{lem}
\begin{proof}
For $\text{Im}\ \xi <0$ we have
\begin{equation}\label{F}
\begin{split}
& F_{nm}(\xi)= \\
& = f_{nm}(\xi)+\frac{v^4f_{n\ N-1}(\xi)f^\sigma_{N N}(\xi)f_{N-1\ m}(\xi)}{1-v^4f_{N-1\ N-1}(\xi)f^\sigma_{N N}(\xi)}= \\
& = \frac{\displaystyle\sum\limits_{l=1}^N a_l^{(n,m)}\displaystyle\prod\limits_{i=1, i \neq l}^N(\xi-x_i)}{\displaystyle\prod\limits_{j=1}^N(\xi-x_j)}+\frac{v^4f^\sigma_{N N}(\xi)\displaystyle\sum\limits_{k,l=1}^N a_k^{(n,N-1)}a_l^{(N-1,m)}\displaystyle\prod\limits_{i=1, i \neq k}^N(\xi-x_i)\displaystyle\prod\limits_{j=1,j \neq l}^N(\xi-x_j)}{\displaystyle\prod\limits_{p=1}^N(\xi-x_p)\left(\displaystyle\prod\limits_{q=1}^N(\xi-x_q)-v^4f^\sigma_{N N}(\xi)\displaystyle\sum\limits_{t=1}^N a_t^{(N-1,N-1)}\displaystyle\prod\limits_{s=1, s \neq t}^N(\xi-x_s)\right)}.
\end{split}
\end{equation}
In the limit $\xi \to x_r$ we obtain the result that $F_{nm}(x_r)$ is finite, as the following asymptotic equality shows\footnote{The needed result can be proved, perhaps in more transparent way, by multiplying the expression for $ F_{nm}(\xi)$ by $\xi - x_r$ and showing that $\lim_{\xi\to x_r-i0} F_{nm}(\xi)(\xi - x_r)=0.$}
\begin{equation*}
\begin{split}
& F_{nm}(\xi) \asymp \\
& \asymp\frac{a_r^{(n,m)}}{\xi-x_r}+\displaystyle\sum\limits_{l=1, l \neq r}^{N}\frac{a_l^{(n,m)}}{x_r-x_l}+\frac{v^4f^\sigma_{N N}(x_r)\left(\displaystyle\frac{a_r^{(n,N-1)}}{\xi-x_r}+\displaystyle\sum\limits_{k=1, k\neq r}^N \frac{a_k^{(n,N-1)}}{x_r-x_k}\right)a_r^{(N-1,m)}\displaystyle\prod\limits_{i=1, i\neq r}^N(x_r-x_i)}{-v^4f^\sigma_{N N}(x_r)a_r^{(N-1,N-1)}\displaystyle\prod\limits_{j=1, j\neq r}^N(x_r-x_j)} \\
& =\text{const}+\frac{1}{\xi-x_r}\left(a_r^{(n,m)}-\frac{a_r^{(n,N-1)}a_r^{(N-1,m)}}{a_r^{(N-1,N-1)}}\right)=\text{const},\ \text{ for}\  \xi\to x_r-i0,
\end{split}
\end{equation*}
because the last term in brackets is equal to zero, cf. \rref{anm}~.
Note that $a_r^{(N-1,N-1)}\neq 0$, since the converse implies $\sin\frac{Nr\pi}{N+1}=0$, which immediately shows that $\frac{Nr}{N+1}=k \in \mathbb{Z}$. For $N>1$, we have $1<k<N$, $Nr=Nk+k$, $1<N(r-k)=k<N$ and finally $0<r-k<1$, a contradiction. Hence the finite limit $\lim_{\xi \to x_j-i0}F_{nm}(\xi)$ exists.
\end{proof}
\noindent We have proved that the functions $F_{nm}$, unlike the functions $f_{nm}$, have no poles at the points $\xi=x_j$ if the function $f^\sigma_{N N}(\xi)$ has nonzero finite limit for $\xi\to x_r-i0.$ From the equation ($\ref{F}$) in the proof we can see that all the singularities of $F_{nm}$ could occur from the singularities of the function $f^\sigma_{N N}$ only, provided that

\begin{equation*}
\lim_{\nu\to 0^+}\left(1-v^4f_{N-1 N-1}(x-i\nu)f^\sigma_{N N}(x-i\nu)\right) \neq 0.
\end{equation*}
So far we used the truncated matrix elements and their Fourier transforms. What we really need for the description of the asymptotic dynamics of our model are the Fourier transforms of the untruncated elements $\mathcal{F}[(\beta_m,e^{itH}\beta_n)]$. In the following, we shall consider our Fourier transforms as the transforms of tempered distributions, see also \cite{RS2}.
\begin{lem}\label{Fourier}
The Fourier transform of the matrix elements of the time evolution operator at $p\in \mathbb{R}$ can be expressed as
\begin{equation*}
\mathcal{F}[(\beta_m,e^{itH}\beta_n)](p)=-\lim_{\nu \to 0^+}\sqrt{\frac{2}{\pi}}\ {\rm Im}\ F_{mn}(p-i\nu).
\end{equation*}\hfill $\triangle$
\end{lem}
\begin{proof}
Denote by $\tilde{f}$ the reflected function to $f$, that is $\tilde{f}(x):=f(-x)$. Recall first from the definition of the Fourier transform that for a locally integrable bounded function $f$ it is $\mathcal{F}[f](p)=\mathcal{F}[\theta f](p)+\mathcal{F}[\tilde{\theta} f](p)$, hence
\begin{equation*}
\mathcal{F}[f](p)=\lim_{\nu\to 0^+}\left(\mathcal{F}[\theta f](p-i\nu)+\mathcal{F}[\tilde{\theta} f](p+i\nu)\right),
\end{equation*}
where the first term can be considered as boundary value of a function analytic in the open lower half-plane, and the second term is boundary value on $\mathbb{R}$ of the function
\begin{equation*}
\frac{1}{\sqrt{2\pi}}\int_{-\infty}^0 e^{-it\xi}f(t)dt=\frac{1}{\sqrt{2\pi}}\int_0^\infty e^{it\xi}f(-t)dt
\end{equation*}
analytic in $\xi$ on whole open upper half of the complex plane $\mathbb{C}$.
From the symmetry of $F_{mn}$ and elementary calculations we have for the Fourier transform integral
\bequ
\begin{split}
\sqrt{2\pi}\mathcal{F}[\tilde{\theta}(t)(\beta_m,e^{itH}\beta_n)](p+i\nu)= &\int_0^\infty e^{it(p+i\nu)}\overline{(\beta_m,e^{itH}\beta_n)}dt=\\
= \overline{\int_0^\infty e^{-it(p-i\nu)}(\beta_m,e^{itH}\beta_n)dt}=& \ \overline{i F_{mn}(p-i\nu)}.
\end{split}
\enqu

The Fourier transform of the matrix element is then
\begin{equation*}
\mathcal{F}[(\beta_m,e^{itH}\beta_n)](p)=
\lim_{\nu\to 0^+}\left(\mathcal{F}[\theta(t)(\beta_m,e^{itH}\beta_n)](p-i\nu)+\mathcal{F}[\tilde{\theta}(t)(\beta_m,e^{itH}\beta_n)](p+i\nu)\right)=
\end{equation*}
\begin{equation*}
=\lim_{\nu \to 0^+}\left(\frac{i}{\sqrt{2\pi}}F_{mn}(p-i\nu)+ \overline{\frac{i}{\sqrt{2\pi}}F_{mn}(p-i\nu)}\right)=
\end{equation*}
\begin{equation*}
=\lim_{\nu \to 0^+}\frac{i}{\sqrt{2\pi}}\left(F_{mn}(p-i\nu)-\overline{F_{mn}(p-i\nu)})\right)=-\lim_{\nu \to 0^+}\sqrt{\frac{2}{\pi}}\ \text{Im}\ F_{mn}(p-i\nu).
\end{equation*}

This concludes the proof.
\end{proof}

Let us choose the parameters of the Hamiltonian $H$, namely $\varepsilon_0$ and the function $\sigma\in L^2(\mathbb{R}^3,d^3x)$ such that for function $\varepsilon(\vec p)$ of the following form
\begin{equation}\label{eps}
\varepsilon(\vec{p}):=a|\vec{p}|^2 ,\ \ a>0,
\end{equation}
we obtain the desired convergence $\lim_{t \to \infty} (\beta_0,e^{itH}\beta_n)=0$ for $n=0,\ldots,N-1$. We also want to obtain a fast rate of convergence. Choose $\sigma \in L^2(\mathbb{R}^3,d^3x)$ such that in the $p$-representation (i.e. for its Fourier image $\mathcal{F}[{\sigma}]$) for some $b>0$ we have

\begin{subequations}
\bequ\label{sig1}
 \mathcal{F}[{\sigma}](\vec{p})=0\  \text{for}\ |\vec{p}|<b,\, \mathcal{F}[{\sigma}](\vec{p})>0\  \text{for all}\ |\vec{p}|>b,\, \text{and}\
 \mathcal{F}[{\sigma}]\in\mathcal{S}(\bR^3)
\enqu
 ($\mathcal{S}(D)$ denotes the Schwartz space of rapidly decreasing smooth functions with supports in the domain $D$).\\
\newline
{\bf\large Example.}\ \ As an example of such a $\sigma$, we could choose, with a $\delta>0,$
\begin{equation}\label{sig}
  \mathcal{F}[{\sigma}](\vec{p}) := \left\{
  \begin{array}{l l}
    0, & \quad |\vec{p}|<b,\\
    C_1\,\omega_\delta(\vec{q}) * (e^{-\alpha |\vec{q}|^2}\theta(|\vec{q}|-b-\delta))(\vec{p}), & \quad |\vec{p}|>b,\\
  \end{array} \right.
\end{equation}

\begin{equation*}
  \omega_\delta(\vec{q}) := \left\{
  \begin{array}{l l}
    \frac{1}{\delta}C_2\, e^{\frac{1}{|\vec{q}|^2/\delta^2-1}} & \quad |\vec{q}|<\delta,\\
    0 & \quad |\vec{q}|>\delta.\\
  \end{array} \right.
\end{equation*}
\end{subequations}
with $C_1,C_2>0$ being constants (the function $\omega_\delta$ is the standard mollifier).\hfill${\bf \triangle}$\\

We define now a measure $\mu$ on $\mathbb{R}$ by
\begin{equation*}
F_\mu(\varepsilon)=\int_{-\infty}^\varepsilon d\mu:=\int_{\varepsilon(\vec{p})<\varepsilon}|\mathcal{F}[{\sigma}](\vec{p})|^2d^3p.
\end{equation*}
  Then $F_\mu(\varepsilon):=\mu((-\infty,\varepsilon))$ is a smooth function and we denote its derivative
  \bequ \label{rhomu}
  \rho_\mu(\varepsilon):= F^{'}_{\mu}(\varepsilon).
  \enqu
  In the case of $\mathcal{F}[{\sigma}](\vec{p})\equiv \mathcal{F}[{\sigma}](|\vec{p}|)$, e.g. for the choice \rref sig~, we obtain
\begin{equation}\label{rhomu1}
F'_\mu(\varepsilon)=\frac{d}{d\varepsilon}\int_{0}^{\sqrt{\frac{\varepsilon}{a}}}2\pi|\vec{p}||\mathcal{F}[{\sigma}](|\vec{p}|)|^2d|\vec{p}|=
\frac{2\pi}{a}\sqrt{\frac{\varepsilon}{a}}\left|\mathcal{F}[{\sigma}]\left(\sqrt{\frac{\varepsilon}{a}}\right)\right|^2.
\end{equation}
 Obviously, $\rho_\mu \in \mathcal{S}(\bR)$ with supp$(\rho_\mu)\subset [ab^2, +\infty)$, $\rho_\mu(\varepsilon)\neq 0$ for $\varepsilon > ab^2$. \\

\noindent We shall investigate now the limit $\lim_{\nu \to 0^+}\text{Im}\ F_{mn}(p-i\nu)$ for all real $p$.  In order to do this, we need to investigate properties of the function $p\mapsto \lim_{\nu\to 0^+}f^\sigma_{N N}(p-i\nu)$ (recall the expression (\ref{Fmn}) for $F_{mn}$). Using the Fourier transformation we have for $f^\sigma_{N N}$ and $\text{Im}\ \xi<0$
\begin{equation}\label{sigg}
\begin{split}
f^{\sigma}_{N N}(\xi)&=(\beta_N(\sigma),R_{H_0}(\xi)\beta_N(\sigma))=\frac{1}{i}\int_0^\infty e^{-it\xi}(\beta_N(\sigma),e^{itH_0}\beta_N(\sigma))dt \\
& =\frac{1}{i}\int_0^\infty e^{-it(\xi+\varepsilon_0)}(\sigma,e^{ith}\sigma)dt=\frac{1}{i}\int_0^\infty \int_{\mathbb{R}^3}|\mathcal{F}[\sigma](\vec{p})|^2e^{-it(\xi+\varepsilon_0-\varepsilon(\vec{p}))}d^3pdt  \\
& =-\int_{\mathbb{R}^3}\frac{|\mathcal{F}[\sigma](\vec{p})|^2}{\xi+\varepsilon_0-\varepsilon(\vec{p})}d^3p.
\end{split}
\end{equation}

Substituting $\rho_\mu$ from $(\ref{rhomu1})$ into ($\ref{sigg}$) we have
\begin{equation}\label{inte}
f^\sigma_{N N}(\xi)=-\int_{ab^2}^\infty\frac{\rho_\mu(\varepsilon)d\varepsilon}{\xi+\varepsilon_0-\varepsilon},
\end{equation}
and using Sokhotski-Plemelj relations \cite[V.3: Example 6]{RS}, we have
\begin{equation}\label{princ}
\lim_{\nu \to 0^+}f^\sigma_{N N}(p-i\nu)=\frac{\pi}{i}\rho_\mu(p+\varepsilon_0)-
P\int_{ab^2}^\infty\frac{\rho_\mu(\varepsilon)d\varepsilon}{p+\varepsilon_0-\varepsilon}\neq 0, \forall p\in\bR.
\end{equation}
For $p+\varepsilon_0\leq ab^2$ the integral in equation \rref inte~ converges and is continuous at $\xi=p\in\bR $,\ $f^\sigma_{N N}(p)>0$, and there\-fore $\lim_{\nu\to 0^+}\text{Im}\  f^\sigma_{N N}(p-i\nu)=0$. Note that $f^\sigma_{N N}(p)$ is analytic for $p < ab^2-\varepsilon_0$. In the case $p+\varepsilon_0 > ab^2$, $\lim_{\nu \to 0^+}\text{Im}\ f^\sigma_{N N}(p-i\nu)=-\pi\rho_\mu(p+\varepsilon_0)<0$. Hence we have $f^\sigma_{N N}(p)\neq 0, \ \forall p\in\bR.$

The principal value integral  can be written as convolution of the tempered distribution $P(\frac{1}{x})$ (because $P(\frac{1}{x})\in\mathcal{S}'(\bR)$, cf. e.g. \cite[V.3:\ Example 6]{RS}) with a function from $\mathcal{S}(\bR)$:
\begin{equation}\label{*}
P\int_{ab^2}^\infty \frac{\rho_\mu(\varepsilon)d\varepsilon}{p+\varepsilon_0-\varepsilon}=P\int_{-\infty}^{p+\varepsilon_0-ab^2} \frac{\rho_{\mu}(p+\varepsilon_0-x)}{x}dx
=P\left(\frac{1}{x}\right)*\left(\rho_\mu(x)\right)(p+\varepsilon_0),
\end{equation}
               ($*$ denotes the convolution) so that by \cite[Theorem IX.4(a)]{RS2} $p \mapsto P(\frac{1}{x})*(\rho_\mu(x))(p+\varepsilon_0)$ is a polynomially bounded $C^\infty(\mathbb{R})$ function. We have just proved:

\begin{lem}   The function $p\mapsto f^\sigma_{N N}(p):=\lim_{\nu \to 0^+}f^\sigma_{N N}(p-i\nu), (p\in\bR)$ is an everywhere nonzero $C^\infty(\mathbb{R})$ function. \hfill $\triangle$
\end{lem}

In Lemma $\ref{s}$ we proved that if $\lim_{\nu \to 0^+}f^\sigma_{N N}(x_j-i\nu) \neq \pm \infty$ and also $\neq 0$, and $v \neq 0$,  then $\lim_{\nu\to 0^+}F_{nm}(p-i\nu)$ is  analytic in neighborhoods of $p=x_j$, $j=1,\ldots, N$. We have proved now, for our choice of parameters, that the functions $F_{nm}$ are indeed regular in neighborhoods of $x_j$, $j=1,\ldots, N$.

We have to investigate, however, the behavior of $\lim_{\nu\to 0^+}$Im\,$F_{nm}(p-i\nu)$ for all $p\in\bR$.  The only singularities of the function $p \mapsto \lim_{\nu \to 0^+}F_{mn}(p-i\nu)$, as can be seen from ($\ref{F}$), could appear in those points $p\in\bR$ where
\begin{equation}\label{menovatel}
1=\lim_{\nu \to 0^+}v^4f^\sigma_{N N}(p-i\nu)f_{N-1 N-1}(p-i\nu), \ \ p \neq x_j.
\end{equation}
We shall choose the parameter $\varepsilon_0$ of the model such that \rref menovatel~ cannot happen.\\

All the $x_j$ are contained in the interval $[-2,2]$. Suppose $\rho_\mu(p+\varepsilon_0)>0$ for $p \in [-2,2]$. It follows that we  need to choose $\varepsilon_0>ab^2+2$. We can rewrite \rref menovatel~\, with the help of  \rref Fmn~, \rref fmn~, \rref princ~\, and \rref *~, for all real $p\neq x_j\, (j=1,2,\dots, N)$ in the form
\begin{equation}\label{spor}
 \prod_{j=1}^N (p-x_j)=
 v^4\left(\frac{\pi}{i}\rho_\mu(p+\varepsilon_0)-P\left(\frac{1}{x}\right)*\rho_\mu(x)(p+\varepsilon_0)\right)
\sum_{l=1}^N a_l^{(N-1,N-1)}\prod_{k\neq l}(p-x_k).
 \end{equation}
This equality could be valid (for $x_j\neq p\in\bR$) only if the imaginary part of the right hand side vanishes,
\begin{equation}\label{valid}
\rho_\mu(p+\varepsilon_0)\sum_{l=1}^N a_l^{(N-1, N-1)}\prod_{k \neq l}(p-x_k)=0.
\end{equation}
Because we are interested in the behavior of $\lim_{\nu\to 0^+}\text{Im}\ F_{nm}(p-i\nu)$, cf. Lemma \ref{Fourier}, because from \rref princ~ it is  $\lim_{\nu\to 0^+}\text{Im}\ f_{N N}^{\sigma}(p-i\nu)=0$ for all $p\leq ab^2-\varepsilon_0$, and because
\bequ
\lim_{\nu\to 0^+} \text{Im}\,f_{N N}^{\sigma}(p-i\nu)=0 \Rightarrow \lim_{\nu\to 0^+}\text{Im}\,F_{nm}(p-i\nu)=0,
\enqu
it is sufficient to investigate validity of \rref menovatel~ for $p\geq ab^2-\varepsilon_0$ only.
Recall that
\begin{equation*}
a_l^{(N-1, N-1)}=-\frac{2}{N+1}\sin^2\left(\frac{Nl\pi}{N+1}\right)<0, \ \ l=1, \ldots, N,
\end{equation*}
and hence all the terms in the sum in the equation ($\ref{valid}$) for $|p|>2$ have the same signature, what implies that the equation is valid only if $\rho_\mu(p+\varepsilon_0)=0$, a contradiction for $p>2$.

It remains to find out when the equation $(\ref{spor})$ holds if $ab^2-\varepsilon_0<p\leq2$.   In this case we have $\rho_\mu(p+\varepsilon_0)>0$, and we obtain for the imaginary and real parts of the equality \rref spor~, respectively,
\begin{subequations}
\begin{equation}\label{Im-spor}
\rho_\mu(p+\varepsilon_0)\displaystyle\sum\limits_{l=1}^N\frac{\sin^2\left(\frac{Nl\pi}{N+1}\right)}{p-2\cos\left(\frac{l\pi}{N+1}\right)}=0,
\end{equation}
\begin{equation}\label{Re-spor}
\frac{2v^4}{N+1}P\int_{ab^2}^\infty\frac{\rho_\mu(\varepsilon)d\varepsilon}{p+\varepsilon_0-\varepsilon}=\frac{1}{\displaystyle\sum\limits_{l=1}^N\frac{\sin^2\left(\frac{Nl\pi}{N+1}\right)}{p-2\cos\left(\frac{l\pi}{N+1}\right)}}.
\end{equation}
\end{subequations}
This, however, is another contradiction: For $p\neq 2\cos\left(\frac{l\pi}{N+1}\right),\, l=1,\ldots, N$, the first equation implies the divergence of the right-hand side in the second equation, but that violates the fact that the left-hand side is $C^\infty(\mathbb{R})$.
We have  the following assertion.
\begin{lem}\label{reg-Fmn}
For any $\varepsilon_0>ab^2+2$ (with possibly one exception), $p\geq ab^2-\varepsilon_0$, and $\varepsilon(\vec{p})$, $\sigma \in L^2(\mathbb{R}^3,d^3x)$ chosen as in ($\ref{eps}$) and ($\ref{sig1}$), the function $p \mapsto \lim_{\nu \to 0^+}F_{mn}(p-i\nu)$ for $p \in \mathbb{R},$
\begin{equation*}
\lim_{\nu \to 0^+}F_{mn}(p-i\nu)=\lim_{\nu \to 0^+}\left(f_{mn}(p)+\frac{v^4f^\sigma_{N N}(p-i\nu)f_{n N-1}(p)f_{m N-1}(p)}{1-v^4f^\sigma_{N N}(p-i\nu)f_{N-1 N-1}(p)}\right)
\end{equation*}
is in \, $C^\infty([ab^2-\varepsilon_0,+\infty))$. \hfill $\triangle$
\end{lem}
\begin{proof}
The only point which remains to prove is the differentiability of $\lim_{\nu\to 0^+}F_{mn}(p-i\nu)$ for $p=p_0:=ab^2-\varepsilon_0$, and the note in brackets. The equation \rref Im-spor~ is for $p=p_0$ fulfilled, because $\rho_\mu(ab^2)=0.$ To avoid singularity of $F_{mn}$ at this point, the equation \rref Re-spor~ should be violated. But the left hand side is (after the substitution $p=p_0$) independent of $\varepsilon_0$ and the right hand side is, after putting in it $p=p_0$, monotonically decreasing function of $\varepsilon_0$ in the interval $\varepsilon_0\in(ab^2+2,+\infty)$. Hence, if for some value of $\varepsilon_0$ the equality in \rref Re-spor~ holds (this would be the exceptional value mentioned in the bracket of the assertion of this Lemma), \rref Re-spor~ would be false for all other values of $\varepsilon_0\in(ab^2+2,+\infty)$. The Lemma is proved.
\end{proof}
\noindent Now the imaginary part of $F_{mn}$ can be expressed as
\begin{equation}\label{sch}
\begin{split}
&\ \ \ \ \ \ \ \ \ \ \ \ \ \ \ \ \ \ \ \ \ \ \ \ \ \ \ \ \ \ \ \ \ \ \ \ {\lim_{\nu \to 0^+}\text{Im}\ F_{mn}(p-i\nu)=}\\
&=\lim_{\nu \to 0^+}\frac{v^4f_{n N-1}(p)f_{m N-1}(p)\ \text{Im}\ f^\sigma_{N N}(p-i\nu)}{\left(1-v^4f_{N-1 N-1}(p)\ \text{Re}\ f^\sigma_{N N}( p-i\nu)\right)^2+\left(v^4f_{N-1 N-1}(p)\ \text{Im}\ f^\sigma_{N N}(p-i\nu)\right)^2}.
\end{split}
\end{equation}
For $p>ab^2-\varepsilon_0$ we have $\lim_{\nu\to 0^+}\text{Im}\ f^\sigma_{N N}(p-i\nu) \neq 0$, which implies that $\text{Im}\ F_{mn}$ from ($\ref{sch}$) is nonzero, whereas for $p\leq ab^2-\varepsilon_0$ is it zero. For $p \to \infty$ we have
\begin{equation*}
\lim_{\nu \to 0^+}\left(\left(1-v^4f_{N-1 N-1}(p)\ \text{Re}\ f^\sigma_{N N}(p-i\nu)\right)^2+\left(v^4f_{N-1 N-1}(p)\ \text{Im}\ f^\sigma_{N N}(p-i\nu)\right)^2\right) = 1+O\left(\frac{1}{p^2}\right),
\end{equation*}
since
\[f_{nm}(p)=O\left(\frac{1}{p}\right).\]
\\
 This combined with the fact that
 $\lim_{\nu \to 0^+}\text{Im}\ f^\sigma_{N N}(p-i\nu)\in \mathcal{S}(\mathbb{R})$ implies, according to Lemma \ref{Fourier},
\bequ
\mathcal{F}[(\beta_m,e^{itH}\beta_n)](p)=-\sqrt\frac{2}{\pi}\lim_{\nu \to 0^+}\text{Im}\ F_{mn}(p-i\nu)\in \mathcal{S}(\mathbb{R}).
\enqu

But the Schwartz set $\mathcal{S}(\mathbb{R})$ of rapidly decreasing smooth functions is invariant with respect to Fourier transform.
Therefore we obtained an "almost exponential decay" in the emission process:
\\
\newline
{\bf\large Theorem.}\ \
{\it In the model described in Section \ref{model}, with any $\varepsilon_0>ab^2+2$ (with possibly one exception), resp. for all
\[\varepsilon_0>2+ab^2+2v^4\int_{ab^2}^\infty \frac{\rho_\mu(\varepsilon) \text{d}\varepsilon}{\varepsilon-ab^2},\]
  with $\varepsilon(\vec{p}):=a|\vec{p}|^2$ , and with ${\sigma}\in\mathcal{S}(\bR^3)$ such that $\mathcal{F}[\sigma](\vec{p})=0\  \text{for}\ |\vec{p}|<b,\, \mathcal{F}[\sigma](\vec{p})>0\  \text{for all}\ |\vec{p}|>b,$ and any fixed $b>0$, the time evolution of the probability of all the $N+1$ spins being turned up (realizing the wanted final state of the spin chain), if initially the Fermi field was in the vacuum state and the first $n$ spins ($n\geq 0$) was turned up, approaches to unity ``almost
exponentially fast'', i.e.
the relation \rref evolution~:
\[(\beta_n,e^{itH}a^*_Na_Ne^{-itH}\beta_n) = 1-o(t^{-m})\]
is valid.}
$\hfill$ $\triangle$

\begin{proof}
The alternative lower bound for $\varepsilon_0$ is obtained from expressions entering into \rref Re-spor~ by putting there $p=ab^2-\varepsilon_0$ and approximating its right hand side  by replacing $1$ for $\sin^2(\dots)$ and $-1$ for $\cos(\dots)$, for any $N$. It is seen then that for all $\varepsilon_0>2+ab^2+2v^4\int_{ab^2}^\infty \frac{\rho_\mu(\varepsilon) \text{d}\varepsilon}{\varepsilon-ab^2},$ the equation   \rref Re-spor~ cannot be fulfilled also for any $p<-2$, and hence for such $\varepsilon_0$ we have $F_{mn}(p):=\lim_{\nu \to 0^+}F_{mn}(p-i\nu)\in C^\infty(\bR).$

The map $t \mapsto (\beta_m,e^{itH}\beta_n) \in \mathcal{S}(\mathbb{R})$ and so $\sum_{n=0}^{N-1}|(\beta_m,e^{itH}\beta_n)|^2\in \mathcal{S}(\mathbb{R})$ and since $a^*_Na_N\beta_n=0$ for $n=0,\ldots, N-1$ and $a^*_Na_N\beta(\phi_j)=\beta(\phi_j)$ we have
\begin{equation*}
(\beta_m,e^{itH}a^*_Na_N e^{-itH}\beta_m)=\sum_{j=0}^\infty (\beta_m,e^{itH}\beta_N(\phi_j))(\beta_N(\phi_j),e^{-itH}\beta_m)
\end{equation*}
\begin{equation*}
=1-\sum_{n=0}^{N-1}(\beta_m,e^{itH}\beta_n)(\beta_n,e^{-itH}\beta_m)=1-\sum_{n=0}^{N-1}|(\beta_m,e^{itH}\beta_n)|^2.
\end{equation*}
Therefore we conclude that for $t \to \infty$ the convergence of the emission probability is given by
\begin{equation*}
1-(\beta_m,e^{itH}a^*_Na_Ne^{-itH}\beta_m)=o\left(\frac{1}{t^n}\right),
\end{equation*}
for all $n \in \mathbb{Z}_+$ and $m=0,1,\ldots,N-1$, and our result is proved.
\end{proof}

\section{Conclusion}\label{conclusion}
The aim of the present work was to construct a solvable QM model of an (effectively) irreversible process mimicking a ``suitable'' measuring apparatus, such that the time evolution of its ``pointer state'' would be noticeably faster than that of the QD-model of Section \ref{domino}.
The  present model of the radiating system constructed in Section \ref{model} consists of a Fermi field and a  spin chain {\em of finite lenght}, which makes it ``more realistic'' than the infinite spin chain of QD described in Section \ref{domino}, (and in more detail in \cite{Bona}), in the sense of eliminating the infiniteness of the spin chain.
With parameters of the model chosen as in the Theorem of Section \ref{dynamics}, the almost exponential decay rate of the emission process of the finite spin chain and convergence of the chain (starting with all but the zeroth spin pointing down) to the state with all spins pointing up was proved.  Because of ``practical impossibility'' (in the sense of Boltzmann ideas) of exact time-reversing of the state after the emission of a fermion, this process can be considered as a model of irreversible behavior. Note that a {\em finite} QD-chain without radiation would evolve almost periodically.

Let us discuss now briefly the possibility to interpret the present model as a model of an apparatus solving partially QMP.
Two distinguished  stationary states of the total system were considered - the state of all spins pointing down with the Fermi field in vacuum state, and the state with all spins up (resulting in the emission of a fermion which escapes to infinity thanks to the choice of the short range interaction and other parameters of the model) and again the vacuum of the Fermi field. If the length N of the chain is ``sufficiently large'', these states could be considered as ``effectively macroscopically distinguishable'' (cf. \cite{Hepp}, but also \cite{leggett}): A possibility of observing some mutual interference of these two ``macro-states'' amounts to a use of an observable represented by an operator with nonzero matrix element between these two states, hence recording a specific $N$-\,spin correlation. The general formalism of QM admits existence of such observables.  The question here is some possibility of realization of a corresponding apparatus (measuring on our long $N$-\,spin chain). Our hypothesis is, that with the growth of the size of the ``large'' system (in our case the size is measured by $N$) the possibility of construction of such an apparatus is less and less probable, so that for ``large enough'' $N$ the states corresponding to the vector $\beta_N:=a_0^*\ldots a_N^*\Omega_0^s\otimes \Omega_0^f$ representing the ``final state'' obtained in the (weak) limit $t\to\infty$, and to the initial vector $\beta_0:=a_0^*\Omega_0^s\otimes \Omega_0^f$ are ``effectively disjoint'' and their mutual interference cannot be observed. This consideration need not contradict the experimental results with QIMDS (``quantum interference between macroscopically distinct states'') described in \cite{leggett}.

 If we interpret the zeroth spin as the microscopic system being measured, the rest of the (finite) spin chain connected to the Fermi field as the macroscopic measurement apparatus and (a part of) environment, the initial state of the zeroth spin with probability $w$ of being in the state ``pointing up'' is almost exponentially fast reflected in the same probability $w$ of the spin chain with all spins being ``turned up'' (which could be interpreted as the probability of a change of the (macroscopic) ``pointer position''). As is well known \cite{Hepp}, it is impossible to describe in any finite time the measurement process corresponding to a ``truly macroscopic change'' ruled by automorphic time evolution of the combined system ``measured microsystem + macroscopic apparatus''. However, if the convergence to the infinite time in a theoretical model is ``fast enough'', in accordance with Hepp \cite{Hepp} it can be concluded  that such a system might provide a model for the effective description of the measurement process.
We admit that the presented model of QMP works as a sort of ``decoherence'', but without substantially changing formalism of QM it would be hardly possible to construct a model of QMP which could not (or need not) be denominated as FAPP.\footnote{The FAPP-principle denotes ``for all practical purposes'' according to J. S. Bell who used it to denote some provisional solution of a problem, c.f. \cite{bell}. We could ask, however, what is not provisional in any human activity, although a conceptually clearer formulation of a solution of QMP then that via decoherence would be more satisfactory also for the present authors.}\\

\paragraph{Acknowledgement}
\noindent The second author would like to acknowledge support of the Comenius University Grant no. UK/495/2011.

\end{document}